\documentclass[11pt]{article}
\usepackage[utf8]{inputenc}
\usepackage{amsmath}
\usepackage{amsthm}
\usepackage{amsfonts}
\usepackage{amssymb}
\usepackage{verbatim}
\usepackage{breqn}
\usepackage{mathtools}
\usepackage{color}
\usepackage{graphicx}

\usepackage{fullpage}

\usepackage{caption}

\usepackage[pdftex,colorlinks=true]{hyperref}
\usepackage[numbers,sort&compress]{natbib}

\newtheorem{problem}{Problem}

\newtheorem{lemma}{Lemma}
\newtheorem{theorem}{Theorem}

\definecolor{darkred}{rgb}{.4,0,0}
\definecolor{fluored}{rgb}{.92,0.11,0.3}
\definecolor{fluoorange}{rgb}{1,.33,0}
\definecolor{fluoblue}{rgb}{0.11,.3,0.92}

\title{Sparse Regression via Range Counting}
\author{%
Jean Cardinal\footnotemark[1]\,\,\footnotemark[2]
\and
Aur\'elien Ooms\footnotemark[3]\,\,\footnotemark[4]%
}

\begin{document}

\renewcommand{\thefootnote}{\fnsymbol{footnote}}
\maketitle
\footnotetext[1]{Universit\'e libre de Bruxelles (ULB), Brussels, Belgium. Email: {\tt jcardin@ulb.ac.be}}
\footnotetext[2]{Supported by the Fonds de la Recherche Scientifique-FNRS under CDR Grant J.0146.18.}
\footnotetext[3]{BARC, University of Copenhagen, Denmark. Email: {\tt aurelien.ooms@di.ku.dk}}
\footnotetext[4]{%
Supported by the VILLUM Foundation grant 16582.
Part of this research was accomplished while the author was a PhD
student at ULB under FRIA Grant 5203818F (FNRS).%
}
\renewcommand{\thefootnote}{\arabic{footnote}}

\begin{abstract}
The sparse regression problem, also known as best subset selection problem, can
be cast as follows: Given a set $S$ of $n$ points in $\mathbb{R}^d$, a point $y\in
\mathbb{R}^d$, and an integer $2 \leq k \leq d$, find an affine combination of at most $k$
points of $S$ that is nearest to $y$.
We describe a $O(n^{k-1} \log^{d-k+2} n)$-time randomized
$(1+\varepsilon)$-approximation algorithm for this problem with \(d\) and
\(\varepsilon\) constant.
This is the first algorithm for this problem running in time $o(n^k)$.
Its running time is similar to the query time of a data structure recently
proposed by Har-Peled, Indyk, and Mahabadi (ICALP'18), while not requiring any
preprocessing.
Up to polylogarithmic factors, it matches a conditional lower bound relying on
a conjecture about affine degeneracy testing.
In the special case where $k = d = O(1)$, we also provide a simple
$O_\delta(n^{d-1+\delta})$-time deterministic exact algorithm, for any \(\delta > 0\).
Finally, we show how to adapt the approximation algorithm for the sparse
linear regression and sparse convex regression problems with the same
running time, up to polylogarithmic factors.
\end{abstract}

\section{Introduction}

Searching for a point in a set that is the closest to a given query point is certainly
among the most fundamental problems in computational geometry.
It motivated the study of crucial concepts such as multidimensional search data structures, Voronoi diagrams, dimensionality reduction, and has immediate applications in the fields of databases and machine learning.
A natural generalization of this problem is to search not only for a single nearest neighbor,
but rather for the nearest {\color{darkred} combination} of a bounded number of points. More precisely,
given an integer $k$ and a query point $y$, we may wish to find an affine combination of $k$ points of the set
that is the nearest to $y$, among all possible such combinations.
This problem has a natural interpretation in terms of sparse approximate solutions to linear systems, and is known as the {\color{darkred} sparse regression}, or {\color{darkred} sparse approximation} problem in the statistics and machine learning literature.
Sparsity is defined here in terms of the $\ell_0$ pseudonorm $\left\lVert .\right\rVert_0$, the number of nonzero components. The {\color{darkred} sparse affine regression} problem can be cast as follows:

\begin{problem}[Sparse affine regression]%
Given a matrix $A\in\mathbb{R}^{d\times n}$, a vector $y\in\mathbb{R}^d$, and
an integer $2\leq k\leq d$, find $x\in\mathbb{R}^n$ minimizing $\left\lVert Ax - y\right\rVert_2$, and such that $\left\lVert x\right\rVert_0\leq k$, and $\sum_{i=1}^n x_i = 1$.
\end{problem}

By interpreting the columns of $A$ as a set of $n$ points in $\mathbb{R}^d$,
the problem can be reformulated in geometric terms as the {\color{darkred} nearest induced flat} problem.
\begin{problem}[Nearest induced flat]\label{prob:nif}
Given a set $S$ of $n$ points in $\mathbb{R}^d$, an additional point
$y\in\mathbb{R}^d$, and an integer $k$ such that $2\leq k\leq d$, find $k$
points of $S$ such that the distance from $y$ to their affine hull is the
smallest.
\end{problem}
Here the distance from a point to a flat is the distance to the closest point on the flat.
Note that if we allow $k=1$ in the above definition, we have the nearest neighbor problem as a special case. We consider the setting in which the dimension $d$ of the ambient space as well as the number $k$ of points in the sought combination are constant, and study the asymptotic complexity of the problem with respect to $n$. As observed recently by Har-Peled, Indyk, and Mahabadi~\cite{HIM18}, the problem is closely related to the classical {\color{darkred} affine degeneracy testing} problem, defined as follows.
\begin{problem}[Affine degeneracy testing]
Given a set $S$ of $n$ points in $\mathbb{R}^d$, decide whether there exists
$d+1$ distinct points of $S$ lying on an affine hyperplane.
\end{problem}
The latter can be cast as deciding whether a point set is in so-called {\color{darkred} general position}, as is often assumed in computational geometry problems.
In the special case $d=2$, the problem is known to be 3SUM-hard~\cite{GO95,BCILOS19}. In general, it is not known whether it can be solved in time $O(n^{d-\delta})$ for some positive $\delta$~\cite{ES95,AC05}.
Supposing it cannot, we directly obtain a conditional lower bound on the complexity of the nearest induced flat problem.
This holds even for approximation algorithms, which return an induced flat whose distance is within some bounded factor of the distance of the actual nearest flat.

\begin{lemma}[Har-Peled, Indyk, and Mahabadi~\cite{HIM18}]
If the nearest induced flat problem can be approximated within any multiplicative factor in time $O(n^{k-1-\delta})$ for some
positive $\delta$, then affine degeneracy testing can be solved in time $O(n^{d-\delta})$.
\end{lemma}
\begin{proof}
Suppose we have an approximation algorithm for the nearest induced flat problem.
Then given an instance of affine degeneracy testing, we can go through every point $y\in S$ and run this algorithm on an instance composed of the set $S\setminus\{y \}$, the point $y$, and $k=d$. The answer to the degeneracy testing instance is positive if and only if for at least one of these instances, the distance to the approximate nearest flat is zero. The running time is $O(n^{d-\delta})$.
\end{proof}

\subsection*{Motivations and previous works}

Sparse regression is a cornerstone computational task in statistics and machine learning, and comes in a number of flavors.
It is also referred to as {\color{darkred} best subset selection} or, more generally, as {\color{darkred} feature selection} problems~\cite{N95,B16}.
In practice, it is often useful to allow for the sparsity constraint by including a penalty term in the objective function, hence writing the problem in a Lagrangian form.
If the $\ell_1$ norm is used instead of the $\ell_0$ norm, this method is known as the LASSO method~\cite{T96}, to which a tremendous amount of research has been dedicated in the past twenty years.
In the celebrated $k$-SVD algorithm for sparse dictionaries design~\cite{AEB06}, the sparse coding stage consists of a number of sparse regression steps. In this context, they are typically carried out using greedy methods such as the matching pursuit algorithm~\cite{MZ93}.
Efficient sparse regression is also at the heart of {\color{darkred} compressed sensing} techniques~\cite{CRT06,D06}.

Aiming at an exhaustive survey of the variants and applications of sparse regression is futile; instead, we refer to Hastie, Tibshirani, and Friedman~\cite{HTF09} (Chapter 3), Miller~\cite{M02}, and references therein. We also point to Bertsimas, Pauphilet, and Van Parys~\cite{BPV19} for a recent survey on practical aspects of sparse regression methods.

The computational complexity of sparse regression problems is also
well-studied~\cite{N95,DMA97,FKT15,FKK16}.
In general, when a solution $x$ is sought that minimizes the number of nonzero
components while being at bounded distance from $y$, the problem is known to be
$\mathsf{NP}$-hard~\cite{N95}.
However, the complexity of the sparse regression problem when
the sparsity constraint $k$ is taken as a fixed parameter has not been
thoroughly characterized.
In particular, no algorithm with running time $o(n^k)$ is known.

Recently, Har-Peled, Indyk, and Mahabadi~\cite{HIM18} showed how to use
approximate nearest neighbor data structures for finding approximate solutions
to the sparse affine regression problem.
They mostly consider the {\color{darkred} online} version of the problem, in which we allow
some preprocessing time, given the input point set $S$, to construct a data
structure, which is then used to answer queries with input $y$.
They also restrict to {\color{darkred} approximate} solutions, in the sense that the
returned solution has distance at most $(1+\varepsilon)$ times larger than the
true nearest neighbor distance for any fixed constant $\varepsilon$.
They show that if there exists a $(1+\varepsilon)$-approximate nearest neighbor
data structure with preprocessing time $S(n,d,\varepsilon)$ and query time
$Q(n,d,\varepsilon)$, then we can preprocess the set $S$ in time
$n^{k-1}S(n,d,\varepsilon)$ and answer regression queries in time
$n^{k-1}Q(n,d,\varepsilon)$.
Plugging in state of the art results on approximate nearest neighbor searching
in fixed dimension~\cite{AMNSW98}, we obtain a preprocessing time of $O(n^k\log
n)$ with query time $O(n^{k-1}\log n)$ for fixed constants $d$ and $\varepsilon$.

Har-Peled et al.~\cite{HIM18} also consider the {\color{darkred} sparse convex regression}
problem, in which the coefficients of the combination are not only
required to sum to one, but must also be nonnegative.
In geometric terms, this is equivalent to searching for the {\color{darkred} nearest
induced simplex}.
They describe a data structure for the sparse convex regression problem having
the same performance as in the affine case, up to a \(O(\log^k n)\) factor.
For \(k=2\),
they also give a $(2+\varepsilon)$-approximation subquadratic-time offline
algorithm. When \(d=O(1)\), the running time of this algorithm can be made
close to linear.

A closely related problem is that of searching for the nearest flat in a
set~\cite{M07,BHZ11,M15}.
This was also studied recently by Agarwal, Rubin, and Sharir~\cite{ARS17}, who
resort to polyhedral approximations of the Euclidean distance to design data
structures for finding an approximate nearest flat in a set. They prove that
given a collection of $n$ $(k-1)$-dimensional flats in $\mathbb{R}^d$, they can
construct a data structure in time $O(n^k\operatorname{polylog}(n))$ time and space that can
be used to answer $(1+\varepsilon)$-approximate nearest flat queries in time
$O(\operatorname{polylog}(n))$. They also consider the achievable space-time tradeoffs.
Clearly, such a data structure can be used for online
sparse affine regression:
We build the structure with all possible ${n\choose k}$ flats induced by the
points of $S$.
This solution has a very large space requirement and does not help in the offline version stated as Problem~\ref{prob:nif}.

In the following, we give an efficient algorithm for Problem~\ref{prob:nif},
and bridge the gap between the trivial upper bound of $O(n^k)$ and the lower
bound given by the affine degeneracy testing problem, without requiring any preprocessing.

\subsection*{Our results}

\paragraph*{Nearest induced line, flat, or hyperplane.}
We prove that the nearest induced flat problem (Problem~\ref{prob:nif}), can be
solved within a $(1+\varepsilon)$ approximation factor for constant $d$ and
$\varepsilon$ in time $O(n^{k-1} \log^{d-k+2} n)$, which matches the lower bound on
affine degeneracy testing, up to polylogarithmic factors.
This is a near-linear improvement on all previous methods.
The running time of our algorithm is also comparable to the query time of the
data structure proposed by Har-Peled et al.~\cite{HIM18}.
The two main tools that are used in our algorithms are on one
hand the approximation of the Euclidean distance by a polyhedral distance, as
is done in Agarwal, Rubin, and Sharir~\cite{ARS17}, and on the other hand a
reduction of the decision version of the
problem to orthogonal range queries.
Note that orthogonal range searching data structures are also used in
Har-Peled et al.~\cite{HIM18}, albeit in a significantly distinct fashion.

In Section~\ref{sec:lines},
as warm-up,
we focus on the special case of
Problem~\ref{prob:nif} in which $d=3$ and $k=2$.
\begin{problem}[Nearest induced line in \(\mathbb{R}^3\)]
\label{prob:nil}
Given a set $S$ of $n$ points in $\mathbb{R}^3$, and an additional point $y$, find two points $a,b\in S$ such that the distance
from $y$ to the line going through $a$ and $b$ is the smallest.
\end{problem}
Our algorithm for this special case already uses all the tools that are subsequently generalized for arbitrary values of $k$ and $d$.
The general algorithm for the nearest induced flat problem is described in Section~\ref{sec:flats}.
In Section~\ref{sec:hyper}, we consider the special case of Problem~\ref{prob:nif} in which $k=d$, which can be cast as the {\color{darkred} nearest induced hyperplane} problem.
\begin{problem}[Nearest induced hyperplane]
\label{prob:nih}
Given a set $S$ of $n$ points in $\mathbb{R}^d$, and an additional point $y$, find $d$ points of $S$ such that the distance
from $y$ to the affine hyperplane spanned by the $d$ points is the smallest.
\end{problem}
For this case, we design an {\color{darkred} exact} algorithm with running time
$O(n^{d-1+\delta})$, for any $\delta > 0$.
The solution solely relies on classical computational geometry
tools, namely point-hyperplane duality and cuttings~\cite{CF90,C93}.

Our algorithms can be adapted to perform {\color{darkred} sparse linear regression}, instead
of sparse affine regression. In the former, we drop the condition that the sum of the coefficients must be equal to one.
This is equivalent to the nearest {\color{darkred} linear} induced \(k\)-flat problem.
It can be solved in the same time as in the affine case.
To see this, realize that the problem is
similar to the nearest induced flat problem where the first vertex is always the origin.
The obtained complexity is the same as the one for the nearest induced flat problem.

\paragraph*{Nearest induced simplex.}
Adapting our algorithm to sparse {\color{darkred} convex} regression, in which $x$ is also required to be positive, is a bit more involved.
Har-Peled et al.~\cite{HIM18} augment their data structure for the nearest induced flat with orthogonal range searching data structures in (\(k+1\))-dimensional space to solve this problem with an extra \(O(\log^k n)\) factor in both the preprocessing and query time.
We show we can perform a similar modification.

The sparse convex regression problem can be cast as the problem of finding the nearest simplex induced by $k$ points of $S$.
\begin{problem}[Nearest induced simplex]
\label{prob:nis}
Given a set $S$ of $n$ points in $\mathbb{R}^d$, an additional point $y$, and an integer $k$ such that $2\leq k\leq d$, find $k$ points of $S$ such that the distance from $y$ to their convex hull is the smallest.
\end{problem}
We prove that this problem can also be approximated within a $(1+\varepsilon)$
approximation factor for constant $d$ and $\varepsilon$ in time
$O(n^{k-1} \log^d n)$, hence with an extra \(O(\log^{k-2} n)\) factor in the
running time compared to the affine case.  This is described in
Section~\ref{sec:convex}.

Our results and the corresponding sections are summarized in Table~\ref{table:results}.

\begin{table}
\centering
\caption{\label{table:results}Results.}
\begin{tabular}{| l | l | c | l |}
\hline
Problem & Section & Approximation & Running Time \\ \hline
Problem~\ref{prob:nil}: Nearest induced line in \(\mathbb{R}^3\) & \S~\ref{sec:lines} & \(1 + \varepsilon\) & \(O_{\varepsilon}(n \log^3 n) \) \\ \hline
Problem~\ref{prob:nif}: Nearest induced flat & \S~\ref{sec:flats} & \(1 + \varepsilon\) & \(O_{d,\varepsilon}(n^{k-1} \log^{d-k+2} n) \) \\ \hline
Problem~\ref{prob:nih}: Nearest induced hyperplane & \S~\ref{sec:hyper} & \(1\) & \(O_{d,\delta}(n^{d-1+\delta}),\,\,\forall \delta > 0 \) \\ \hline
Problem~\ref{prob:nis}: Nearest induced simplex & \S~\ref{sec:convex} & \(1 + \varepsilon\) & \(O_{d,\varepsilon}(n^{k-1} \log^d n) \) \\ \hline
\end{tabular}
\end{table}

\section{A $(1+\varepsilon)$-approximation algorithm for the nearest induced line problem in $\mathbb{R}^3$}
\label{sec:lines}

We first consider the nearest induced line problem (Problem~\ref{prob:nil}).
We describe a near-linear time algorithm that returns a $(1+\varepsilon)$-approximation to the nearest induced line in $\mathbb{R}^3$, that is,
a line at distance at most $(1+\varepsilon)$ times larger than the distance to the nearest line.

\begin{theorem}
\label{thm:lines}
For any positive constant $\varepsilon$, there is a randomized
$(1+\varepsilon)$-approximation algorithm for the nearest induced line
problem in $\mathbb{R}^3$ running in time $O_{\varepsilon}(n\log^3 n)$ with high
probability.
\end{theorem}

\paragraph*{$(1+\varepsilon)$-approximation via polyhedral distances.}
The proof uses the following result due to Dudley~\cite{D74}, that is also a
major ingredient in the design of the data structure described by Agarwal,
Rubin, and Sharir~\cite{ARS17}.
The {\color{darkred} polyhedral distance} $d_Q(y,v)$ between two points $y$ and $v$ with
respect to a polyhedron $Q$ centered on the origin is the smallest $\lambda$
such that the dilation $\lambda Q$ of $Q$ centered on $y$ contains $v$, hence
such that $v\in y+\lambda Q$.

\begin{lemma}[Dudley~\cite{D74}]
\label{lem:polyd}
For any positive integer $d$ and positive real $\varepsilon$, there exists a $d$-dimensional polyhedron $Q$ of size $O(1/\varepsilon^{(d-1)/2})$ such that for every $y,v\in \mathbb{R}^d$:
$$
\left\lVert y-v\right\rVert_2 \leq d_Q(y,v) \leq (1+\varepsilon)\cdot \left\lVert y-v\right\rVert_2.
$$
\end{lemma}

\begin{proof}[Proof of Theorem~\ref{thm:lines}]
We reduce the problem to a simpler counting problem in two steps.

\paragraph*{Edge-shooting.}
We use Lemma~\ref{lem:polyd} for $d=3$. We give an exact algorithm for computing
the nearest induced line with respect to a polyhedral distance $d_Q$, where
$Q$ is defined from $\varepsilon$ as in Lemma~\ref{lem:polyd}.
Given a polyhedron $Q$, one can turn it into a simplicial polyhedron by
triangulating it.
Therefore, for constant values of $\varepsilon$, this reduces the problem to a
constant number of instances of the {\color{darkred} edge-shooting problem}, defined as
follows: Given an edge $e$ of $Q$, find the smallest value $\lambda$ such that
$y+\lambda e$ intersects a line through two points of $S$.
We iterate this for all edges of $Q$, and pick the minimum value.
This is exactly the polyhedral distance from $y$ to its nearest induced line.

\paragraph*{Binary search.}

Using a randomized binary search procedure, we reduce the edge-shooting
problem to a simpler {\color{darkred} counting problem}, defined as follows: given the
triangle $\Delta$ defined as the convex hull of $y$ and $y+\lambda e$, count
how many pairs of points $a,b\in S$ are such that the line $\ell (a,b)$ through
them intersects $\Delta$.
Suppose there exists a procedure for solving this problem.
We can use this procedure to solve the edge-shooting problem as follows.

First initialize $\lambda$ to some upper bound on the distance (for instance,
initialize \(\lambda\) to the distance to the closest data point \(p \in S\):
\(\lambda = \min_{p \in S} \left\lVert p - y\right\rVert_2\)).
Then count how many lines $\ell (a,b)$ intersects $\Delta$, using the procedure.
If there is only one, then return this value of $\lambda$.
Otherwise, pick one such line uniformly at random, compute the value $\lambda$
such that this line intersects $y+\lambda e$.
Then iterate the previous steps starting with this new value of $\lambda$.
Since we picked the line at random, and since there are \(O(n^2)\) such lines
at the beginning of the search, the number of iterations of this binary search
is $O(\log n)$ with high probability.

We therefore reduced the nearest induced line problem to
$O(\varepsilon^{-1} \log n)$ instances of the counting problem.

\paragraph*{Orthogonal range counting queries.}

\begin{figure}
\begin{center}
\includegraphics[page=1,scale=.8]{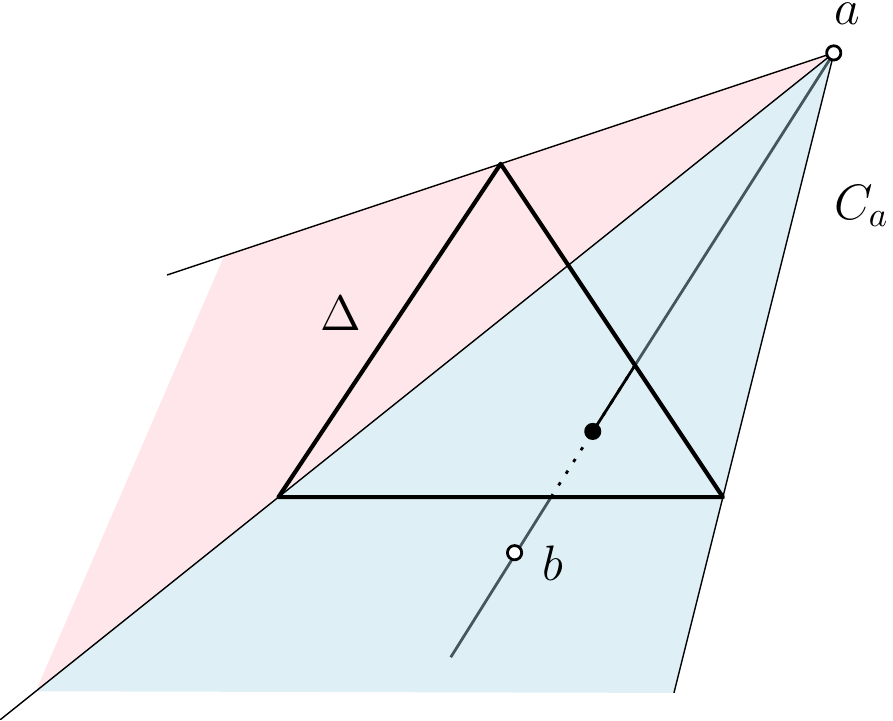}
\end{center}
\caption{\label{fig:cone}The cone $C_a$.}
\end{figure}

Data structures for {\color{darkred} orthogonal range counting queries} store a set of
points in $\mathbb{R}^g$ in such a way that the number of points in a given
$g$-rectangle (cartesian product of $g$ intervals) can be returned quickly.
Known data structures for orthogonal range counting queries in $\mathbb{R}^g$
require $O(n \log^{g-1} n)$ preprocessing time and can answer queries in
$O(\log^{g-1} n)$ time~\cite{Wil85,C88}.
Note that the actual coordinates of the points do not matter: We only need to
know the order of their projections on each axis.
We now show how to solve the counting problem using a data structure for
orthogonal range queries in $\mathbb{R}^3$.

Let us fix the triangle $\Delta$ and a point $a\in\mathbb{R}^3$, and consider
the locus of points $b\in\mathbb{R}^3$ such that the line $\ell (a,b)$
intersects $\Delta$.
This is a double simplicial cone with apex $a$ and whose boundary contains the
boundary of $\Delta$.
This double cone is bounded by three planes, one for each edge of $\Delta$.
In fact, we will only consider one of the two cones, because $\ell (a,b)$
intersects $\Delta$ if and only if either $b$ is contained in the cone of apex
$a$, or $a$ is contained in the cone of apex $b$. Let us call $C_a$ the cone of
apex $a$. This is illustrated on Figure~\ref{fig:cone}.

\begin{figure}
\begin{center}
\includegraphics[page=6,scale=.8]{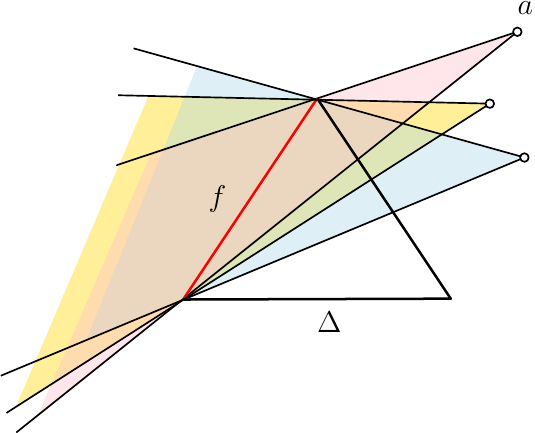}
\end{center}
\caption{\label{fig:order}%
The order of the points defined by the planes containing an edge $f$ of
$\Delta$.}
\end{figure}

Let us consider one edge $f$ of $\Delta$ and all the planes containing $f$.
These planes induce a circular order on the points of $S$, which is the order in which they are met by a plane rotating around the supporting line of $f$.
This is illustrated on Figure~\ref{fig:order}.
Now let us denote by $H_f$ the plane containing $a$ and $f$ and by $H^+_f$ the halfspace bounded by $H_f$ and containing $\Delta$.
The set of points of $S$ contained in $H^+_f$ is an interval in the circular order mentioned above.
Hence the set of points contained in $C_a$ is the intersection of three intervals in the three circular orders defined by the three edges of $\Delta$.

Now we can use an orthogonal range counting data structure for storing the points of $S$ with coordinates corresponding to their ranks in each of the three orders induced by the three edges of $\Delta$. This requires sorting them three times, in time $O(n\log n)$.
Then for each point $a\in S$, we count the number of points $b$ in the cone $C_a$ by querying the data structure.
Note that the circularity of the order can be easily handled by doubling every point.
We preprocess $S$ in time $O(n\log n)$ and answer each of the $n$ queries in time $O(\log^2 n)$.
Hence overall, this takes time $O(n\log^2 n)$.

This can be combined with the previous reductions provided we can choose a line intersecting $\Delta$ uniformly at random within that time bound.
This is achieved by first choosing $a$ with probability proportional to the number of points $b$ such that $\ell (a,b)\cap \Delta\not=\emptyset$.
Then we can pick a point $b$ uniformly at random in this set in linear time.

Combining with the previous reductions, we therefore obtain an approximation
algorithm running in time $O_{\varepsilon}(n\log^3 n)$ for the nearest induced
line problem in $\mathbb{R}^3$. This completes the proof of
Theorem~\ref{thm:lines}.
\end{proof}

\section{A $(1+\varepsilon)$-approximation algorithm for the nearest induced flat problem}
\label{sec:flats}

This section is dedicated to proving our main result in full generality.
We provide an efficient approximation algorithm for the nearest induced flat problem (Problem~\ref{prob:nif}).

We use the following notations: $\operatorname{aff}(X)$ denotes the affine hull of the set $X$ and $\operatorname{conv}(X)$ denotes its convex hull.
The set $\{\,1,2,\ldots ,n\,\}$ is denoted by $[n]$.

\begin{theorem}
\label{thm:flats}
For any constant positive real $\varepsilon >0$ and constant positive integers $d$ and
$k$, there is a randomized $(1+\varepsilon)$-approximation algorithm for the
nearest induced flat problem in $\mathbb{R}^d$ running in time
$O_{\varepsilon}(n^{k-1}\log^{d-k+2} n)$ with high probability.
\end{theorem}

\begin{proof}
The algorithm is a generalization of the one in the previous section, in which the point $a$ is replaced by a set composed of $k-1$ points $a_1,a_2,\ldots,a_{k-1}$, and the edge $e$ is now a (simplicial) $(d-k)$-face of $Q$.
Given a $k-1$-tuple of points $a_1,a_2,\ldots ,a_{k-1}$, we characterize the locus of points $a_k$ such that the affine hull of the points $a_1,a_2,\ldots ,a_k$ intersects the convex hull of $y$ and $y+\lambda e$.
These hyperplanes are again such that counting all such points can be done using orthogonal range queries.
More precisely, we perform the following steps.

\paragraph*{$(1+\varepsilon)$-approximation and binary search.}
From Dudley's result, there exists a polyhedron of size
$O(1/\varepsilon^{(d-1)/2})$ such that the induced polyhedral distance
$d_Q(.,.)$ is a $(1+\varepsilon)$-approximation of the Euclidean distance.
We know that the distance $d_Q$ from the point $y$ to the nearest induced flat
is attained at a point lying on a $(d-k)$-face of $y+\lambda Q$.
We can therefore perform the same procedure as in the previous case, except
that we now shoot a $(d-k)$-face $e$ of $Q$, instead of an edge, in the same
way as is done in Agarwal, Rubin, Sharir~\cite{ARS17}.
$\Delta$ still denotes the convex hull of $y$ and $y+\lambda e$,
which generalizes to a $(d-k+1)$-simplex.
The binary search procedure generalizes easily: start with a large enough
\(\lambda\), if there is more than one flat \(\operatorname{aff}(\{\, a_1, a_2, \ldots,
a_{k}\,\})\) intersecting \(\Delta = \operatorname{conv}(\{\, y, y+ \lambda e\,\})\),
pick one such flat uniformly at random, and compute the value $\lambda$
such that this flat intersects \(\Delta\). There are only \(O(n^k)\) such flats
at the beginning of the search, hence a search takes \(O(\log n)\) steps with
high probability.
We can therefore reduce the problem to $O(\varepsilon^{(1-d)/2} \log n)$
instances of the following {\color{darkred} counting problem}: given a $(d-k+1)$-simplex
$\Delta$, count the number of $k$-tuples of points $a_1,a_2,\ldots ,a_k\in S$
whose affine hull $\operatorname{aff} (a_1,a_2,\ldots ,a_k)$ intersects $\Delta$.

\paragraph*{An intersection condition.}
We first make a simple observation that characterizes such $k$-tuples.
Let $A$ be a set of $k$ points $\{a_1,a_2,\ldots ,a_k\}$, and let $B
=\{b_1,b_2,\ldots ,b_{d-k+2}\}$ be the set of vertices of $\Delta$.
We assume without loss of generality that the points of $A$ together with the
vertices of $\Delta$ are in general position.
We define $d-k+2$ hyperplanes \(H_i = \operatorname{aff}(A \cup B \setminus \{\,b_i,
a_k\,\}), i \in [d-k+2]\).
We then let \(H_i^+\) be the halfspace supported by \(H_i\) that
contains \(b_i\), and $H_i^-$ the halfspace that does not contain $b_i$.

\begin{lemma}
\label{lem:range}
\[
\operatorname{aff}(A) \cap \Delta \neq \emptyset
\iff
a_k \in \left(
\left(\bigcap_{i=1}^{d-k+2} H_i^+\right)
\cup
\left(\bigcap_{i=1}^{d-k+2} H_i^-\right)
\right).
\]
\end{lemma}
\begin{proof}
($\Rightarrow$) Suppose that $a_k\not\in (\bigcap_i H_i^+) \cup (\bigcap_i H_i^-)$.
Hence there exists $i\in [d-k+2]$ such that $a_k\in H_i^-$, and $j\in [d-k+2]$ such that $a_k\in H_j^+$.
We show that $\operatorname{aff}(A)\cap\Delta =\emptyset$.
Let us consider the intersection of the two halfspaces $H^-_i$ and $H^+_j$ with the $(k-1)$-dimensional subspace $\operatorname{aff} (A)$.
In this subspace, both halfspaces have the points $a_1,a_2,\ldots ,a_{k-1}$ on their boundary, and contain $a_k$.
Hence it must be the case that $H^-_i\cap \operatorname{aff}(A) = H^+_j\cap \operatorname{aff}(A)$. Therefore, every point $p\in\operatorname{aff} (A)$ either lies
in $H^-_i$, or in $H^-_j$. In both cases, it is separated from $\Delta$ by a hyperplane, and $p\not\in\Delta$.

($\Leftarrow$) Suppose that $\operatorname{aff} (A)\cap\Delta = \emptyset$.
We now show that there exists $i\in [d-k+2]$ such that $a_k\in H_i^-$, and $j\in [d-k+2]$ such that $a_k\in H_j^+$.
Since both $\operatorname{aff}(A)$ and $\Delta$ are convex sets, if $\operatorname{aff} (A)\cap\Delta = \emptyset$ then there exists a hyperplane $H$ containing $\operatorname{aff} (A)$ and having $\Delta$ on one side.
Since the points of $A$ are affinely independent, it can then be rotated to contain all points of $A$ except $a_k$, and separate $a_k$ from $\Delta$.
After this, it has $d-(k-1)$ degrees of freedom left, and can be further rotated to contain a whole $(d-k)$-face of $\Delta$, while still separating $\Delta$ from $a_k$.
For some $i\in [d-k+2]$, this is now the hyperplane $H_i$ that separates some vertex $b_i$ from $a_k$, and $a_k\in H_i^-$.

Similarly, the same hyperplane $H$ can instead be rotated in order to contain all points of $A$ except $a_k$, and have $a_k$ and $\Delta$ this time on the same side.
It can then be further rotated to contain a $(d-k)$-face of $\Delta$, while still having $\Delta$ and $a_k$ on the same side.
Now for some $j\in [d-k+2]$, this is now the hyperplane $H_j$ that has $b_j$ and $a_k$ on the same side, and $a_k\in H_j^+$.
\end{proof}

Note that for the case $k=2$ and $d=3$ the set $(\bigcap_i H_i^+) \cup (\bigcap_i
H_i^-)$ is the double cone of apex $a$; the lower part $(\bigcap_i
H_i^+)$ is the cone $C_a$ in Figure~\ref{fig:cone}.
The case where $k=3$ and $d=3$ is illustrated on Figure~\ref{fig:condition}.

\begin{figure}
\begin{center}
\includegraphics[page=5,scale=.8]{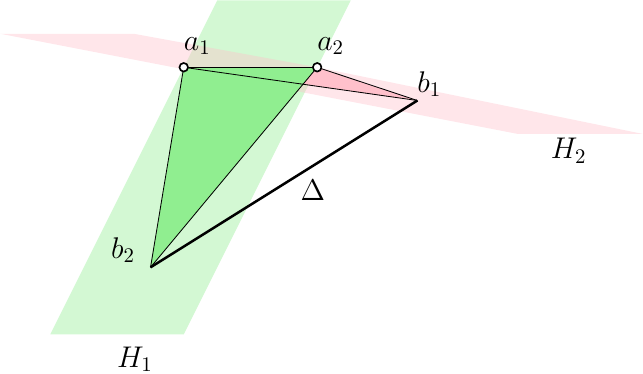}
\end{center}
\caption{\label{fig:condition}Illustration of Lemma~\ref{lem:range} in the case $k=d=3$. The plane through $a_1,a_2,a_3$ intersects the line segment $\Delta$ if and only if $a_3$ is located either above or below the two planes $H_1,H_2$.}
\end{figure}

\paragraph*{Reduction to orthogonal range queries.}

We now show that in perfect analogy with the previous section, we can solve the counting problem efficiently using an orthogonal range counting data structure.

Consider a vertex $b_i$ of $\Delta$ and a $(k-2)$-subset $T$ of points of $S$,
denoted by $T=\{ a_1,a_2,\ldots ,a_{k-2}\}$.
Let us denote by $f$ the facet of $\Delta$ that is induced by the vertices
$b_j$ such that $j \neq i$.
Now consider the hyperplane containing $f$ together with $T$, and one additional point $p$ of $S$.
These hyperplanes all contain $\operatorname{aff} (f \cup T)$, which is a $(d-2)$-flat.
Let us consider the unit normal vectors to these hyperplanes centered on some point contained in this $(d-2)$-flat.
These vectors lie in the orthogonal flat of dimension $d-(d-2)=2$, hence in a plane.
Therefore, they induce a circular order on the points of $S$.
Hence for a fixed set of $k-2$ points of $S$ and a fixed facet $f$ of $\Delta$,
we can assign a rank to each other point of $S$.
These will play the role of the coordinates of the points in the range counting data structure.

We now observe that counting the number of $k$-tuples whose affine hull intersects $\Delta$ amounts to orthogonal range counting with respect to these coordinates.
Indeed, fix the first $(k-2)$-subset of points $T=\{a_1,a_2,\ldots ,a_{k-2}\}$,
and compute the rank of each other point of $S$ with respect to the circular
order of the hyperplanes defined above, around each facet $f$ of $\Delta$.
Now consider a $(k-1)$th point $a_{k-1}$.
From Lemma~\ref{lem:range}, all points $a_k$ contained in the range $(\bigcap_i H_i^+) \cup (\bigcap_i H_i^-)$ are such that $\operatorname{aff} ({a_1,a_2,\ldots ,a_k})$ intersects $\Delta$.
But this range is the union of two $(d-k+2)$-rectangles in the space of coordinates that we defined.
The coordinates of these two $(d-k+2)$-rectangles are defined by the coordinates of $a_{k-1}$.
We can therefore set up a new orthogonal range counting data structure for each
$(k-2)$-subset $T$, and perform $2n$ queries in it,
two for each additional point $a_{k-1}\in S$.

We can now outline our algorithm for solving the counting problem:
\begin{enumerate}
\item \label{step:k-2} For each $(k-2)$-subset $T$ of points $a_1,a_2,\ldots
,a_{k-2}$ in ${S\choose k-2}$:
\begin{enumerate}
\item For each vertex $b_i$ of $\Delta$, compute the rank of each point of $S$
with respect to the hyperplanes containing $f=\operatorname{conv} (\{b_j:j\not= i\})$ and
$T$.
\item \label{step:DS} Build a $(d-k+2)$-dimensional range counting data
structure on $S$ using these ranks as coordinates.
\item For each other point $a_{k-1}\in S$:
\begin{enumerate}
\item \label{step:queries} Perform two range counting queries using the
rectangular ranges corresponding to $\bigcap_i H^+_i$ and $\bigcap_i H^-_i$,
respectively.
\end{enumerate}
\item Return the sum of the values returned by the range counting queries.
\end{enumerate}
\end{enumerate}

Note that there are a few additional technicalities which we have to take care
of.
First, the orders defined by the hyperplanes are circular, hence we are really
performing range queries on a torus.
This can be easily fixed, as mentioned previously, by doubling each point.
Then we have to make sure to avoid double counting, since any permutation of
the \(a_i\) in the enumeration of \(k\)-tuples yields the same set \(A\), and
hence, the same flat \(\operatorname{aff}(A)\).
(Note that in \S~\ref{sec:lines} we avoided double counting by observing
that only one of \(a \in C_b\) and \(b \in C_a\) can be true.)
This only affects the counting problem and is not problematic if we consider
\emph{ordered} subsets \(T\); it causes each intersecting flat to be
counted exactly \(k!\) times.\footnote{%
Enumerating each subset \(T\) exactly once as (\(k-2\))-tuples
in lexicographic order and only constructing the orthogonal range
searching data structure on the points of \(S\) that come after \(a_{k-2}\)
reduces this overcounting to \(2\) times per flat. In our case, this is
unnecessary since \(k\) is constant.
}
The termination condition for the binary search can be changed to when the
range count is \(k!\) and the sampling method for finding a uniform random
binary search pivot is unaffected since each candidate flat is represented an
equal number of times.

As for the running time analysis, step~\ref{step:DS} costs $O(n\log^{d-k+1}
n)$, while step~\ref{step:queries} costs $O(\log^{d-k+1} n)$ and is
repeated $n - k + 2$ times, hence costs $O(n\log^{d-k+1}n)$ overall as
well~\cite{Wil85,C88}.
These are multiplied by the number of iterations of the main loop, yielding
a complexity of $O(n^{k-1}\log^{d-k+1}n)$ for the counting procedure.

Finally, this counting procedure can be combined with the binary search
procedure provided we can choose a flat intersecting $\Delta$ uniformly at
random within that time bound.
This is achieved by first choosing a set prefix \(\{\, a_1, a_2, \ldots,
a_{k-1}\,\} \in {S \choose k-1}\) with probability proportional to the
number of points $a_k \in S$ such that
$\operatorname{aff}(\{\, a_1, a_2, \ldots, a_{k}\,\}) \cap \Delta \neq \emptyset$.
Then we can pick a point $a_k$ uniformly at random in this set in linear time.
Multiplying by the number of edge-shooting problems we have to solve, the
counting procedure is invoked $O(\varepsilon^{(1-d)/2} \log n)$ times, yielding
the announced running time.
\end{proof}

\section{An exact algorithm for the nearest induced hyperplane problem}
\label{sec:hyper}

In this section we consider the special case $k=d$, the nearest induced hyperplane problem (Problem~\ref{prob:nih}).
The previous result gives us a randomized (\(1+\varepsilon\))-approximation
algorithm running in time $O_{\varepsilon}(n^{d-1}\log^2 n)$ for this problem.
We describe a simple deterministic $O(n^{d-1+\delta})$-time exact algorithm
using only standard tools from computational geometry.

\begin{theorem}
\label{thm:hyper}
The nearest induced hyperplane problem can be solved in deterministic time $O(n^{d-1+\delta})$ for any $\delta>0$.
\end{theorem}

\begin{figure}
\begin{center}
\includegraphics[page=2, scale=.8]{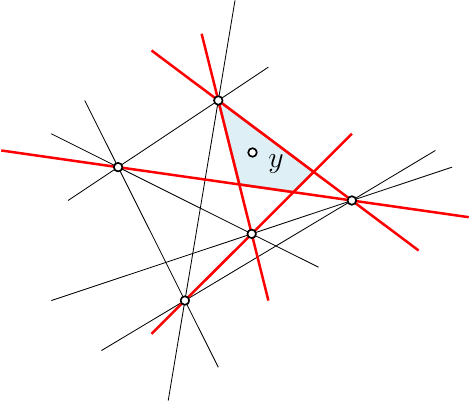}
\end{center}
\caption{\label{fig:zone}The candidate nearest hyperplanes.}
\end{figure}

The first tool we need is point-hyperplane duality.
Let $\bar{H}$ be the hyperplane arrangement that is dual to $S$, in which each point of $S$ is now a hyperplane.
Note that every vertex of this arrangement is the dual of a hyperplane induced by $d$ points of $S$.

Unfortunately, while some dualities preserve vertical distances, there does not
exist a duality that preserves euclidean distances.
To overcome this obstacle, we make a topological observation.
Recall that the {\color{darkred} zone} of a hyperplane $h$ in an arrangement $\bar{H}$ (not
including $h$) is the union of the $d$-cells of $\bar{H}$ intersected by $h$.
Similarily, we define the \emph{refined zone} of a hyperplane $h$ in an arrangement $\bar{H}$ (not
including $h$) to be the union of the $d$-simplices of the bottom-vertex
decomposition of $\bar{H}$ intersected by $h$.

\begin{lemma}
\label{lem:refzone}
Let $\bar{H}$ be the hyperplane arrangement that is dual to $S$, and $\bar{y}$
the hyperplane dual to the point $y$.
The induced hyperplane that is nearest to $y$ corresponds to a vertex of the
{\color{darkred} refined zone} of $\bar{y}$ in the arrangement $\bar{H}$.
\end{lemma}
\begin{proof}
Consider the arrangement of all ${n\choose k}$ hyperplanes induced by subsets
of $k$ points in $S$. Then clearly, the induced hyperplane nearest to $y$ must be one
of the hyperplanes bounding the cell of this arrangement that contains $y$ (see
Figure~\ref{fig:zone} for an illustration with $d=2$).  Consider a rectilinear
motion of $y$ towards this nearest hyperplane. In the dual arrangement
$\bar{H}$, this corresponds to a continuous motion of the hyperplane $\bar{y}$
that at some point hits a vertex of the arrangement. Because it is the first
vertex that is hit, it must belong to a cell of the bottom vertex
decomposition of $\bar{H}$ that $\bar{y}$ intersects, hence to the refined zone
of $\bar{y}$.
\end{proof}
We refer to chapter 28 of the Handbook of Discrete and Computational
Geometry~\cite{book17} for background on hyperplane arrangements and their
decompositions.

The second tool is an upper bound on the complexity of a zone in an arrangement~\cite{ESS93}.
The \emph{complexity of a zone} is the sum of the complexities of its cells, and the
complexity of a cell is the number of faces of the cell (vertices, edges,
\dots).
The upper bound is as follows:

\begin{theorem}[Zone Theorem~\cite{ESS93}]
\label{thm:zone}
The complexity of a zone in an arrangement of $n$ hyperplanes in $\mathbb{R}^d$ is $O(n^{d-1})$.
\end{theorem}
In particular, this result gives an upper bound of \(O(n^{d-1})\) vertices for
a given zone. Since the complexity of a refined zone is not more than the
complexity of the corresponding zone, this bound also holds for the complexity
of a given refined zone.

The third tool is Chazelle's efficient construction of cuttings~\cite{C93}.
A \emph{cutting} of \(\mathbb{R}^d\) is a partition of \(\mathbb{R}^d\) into
disjoint regions. Given a set of hyperplanes \(H\) in \(\mathbb{R}^d\), a
\emph{\(\frac{1}{r}\)-cutting} for \(H\) is a cutting of \(\mathbb{R}^d\)
such that each region is intersected by no more than \(\frac{| H |}{r}\)
hyperplanes in \(H\).
In particular, we are interested in Chazelle's construction when \(r\) is
constant. In that case, only a single step of his construction is necessary and
yields regions that are the simplices of the bottom-vertex decomposition of
some subset of \(H\).
\begin{theorem}[{Chazelle~\cite[Theorem 3.3]{C93}}]
\label{thm:cutting}
Given a set \(H\) of \(n\) hyperplanes in \(\mathbb{R}^d\), for any real
constant parameter \(r > 1\), we can construct a \(\frac{1}{r}\)-cutting for
those hyperplanes consisting of the \(O(r^d)\) simplices of the bottom-vertex
decomposition of some subset of \(H\) in \(O(n)\) time.
\end{theorem}
More details on cuttings can be found in chapters 40 and 44 of the Handbook~\cite{book17}.

\begin{lemma}
\label{lem:cutzone}
For any positive constant $\delta$, given a hyperplane \(h\) and an arrangement of hyperplanes \(\bar{H}\) in \(\mathbb{R}^d\),
the vertices of the refined zone of \(h\) in \(\bar{H}\) can be computed in time $O(n^{d-1+\delta})$.
\end{lemma}
\begin{proof}
Using Theorem~\ref{thm:cutting} with some constant $r$, we
construct, in linear time, a $\frac 1r$-cutting of the arrangement consisting
of $O(r^d)$ simplicial cells whose vertices are vertices of \(\bar{H}\).
To find the vertices of the refined zone, we only need to look at those cells
that are intersected by $\bar{y}$. If such a cell is not intersected by any
hyperplane of \(\bar{H}\) then its vertices are part of the refined zone of
\(\bar{y}\). Otherwise, we recurse on the hyperplanes intersecting that cell.
From Theorem~\ref{thm:zone}, there are at most $O(r^{d-1})$ such cells.
The overall running time for the construction is therefore:
$$
T(n)\leq O\left(r^{d-1}\right) T\left(\frac{n}{r}\right) + O(n).
$$
For all constant \(\delta > 0\), we can choose a sufficiently large constant \(r\),
such that \(T(n) = O(n^{d-1+\delta})\), as claimed.
\end{proof}

\begin{proof}[Proof of Theorem~\ref{thm:hyper}]
From Lemma~\ref{lem:cutzone}, we find the vertices of the refined zone of
\(\bar{y}\) in the arrangement \(\bar{H}\) in time $O(n^{d-1+\delta})$.
Then we compute the distance from $y$ to each of the induced hyperplanes
corresponding to vertices of the refined zone in time $O(n^{d-1})$. From
Lemma~\ref{lem:refzone}, one of them must be the nearest.
\end{proof}

\section{A $(1+\varepsilon)$-approximation algorithm for the nearest induced simplex problem}
\label{sec:convex}

We now consider the nearest induced simplex problem (Problem~\ref{prob:nis}).
The algorithm described in \S\ref{sec:lines} for the case
$k=2$ and $d=3$ can be adapted to work for this problem.

As in \S\ref{sec:lines},
consider the computation of the nearest induced segment
under some polyhedral distance \(d_Q\) approximating the Euclidean distance.
The reduction from this computation to edge-shooting still works with some
minor tweak: if we shoot edges to find the nearest
induced segment under \(d_Q\), we may miss some of the segments. Fortunately,
the points of these missed segments that are nearest to our query point under
\(d_Q\) must be endpoints of those segments. We can take those into account by
comparing the nearest segment found by edge-shooting to the nearest neighbor,
found in linear time.
As before, edge-shooting is reduced to a counting problem.

Referring to the proof of
Theorem~\ref{thm:lines} and Figure~\ref{fig:cone},
the analogue of the counting problem in \S\ref{sec:lines} for the nearest induced
segment problem amounts to searching for the points $b$ lying in the
intersection of the cone $C_a$ with the halfspace bounded by $\operatorname{aff}(\Delta)$
that does not contain $a$.
In dimension $d$, the affine hull of $\Delta$ is a hyperplane, and we restrict $b$ to lie on one side of this hyperplane.

We therefore get a $(1+\varepsilon)$-approximation $O(n \log^d n)$-time
algorithm for the {\color{darkred} nearest induced segment} problem in any fixed dimension
\(d\).
This compares again favorably with the $(2+\varepsilon)$-approximation $O(n\log n)$-time algorithm proposed by Har-Peled et al.~\cite{HIM18}.\\

We generalize this to arbitrary values of $k$ and prove the following result.
\begin{theorem}
\label{thm:convex}
For any constant positive real $\varepsilon >0$ and constant positive integers $d$ and
$k$, there is a randomized $(1+\varepsilon)$-approximation algorithm for the
nearest induced simplex problem in $\mathbb{R}^d$ running in time
$O(n^{k-1} \log^d n)$ with high probability.
\end{theorem}

Again, we compute the nearest induced simplex under some polyhedral distance
\(d_Q\). As in the case \(k=2\),
(\(d-k\))-face-shooting can be adapted to take care of missed simplices:
for each \(2 \leq k' \leq k\),
shoot (\(d-k'\))-faces of \(Q\) to find the nearest (\(k'-1\))-simplex.
For \(k'=1\), find the nearest neighbor in linear time.
For any (\(k-1\))-simplex,
let \(0 \leq k' \leq k\) be the smallest natural number such that no (\(d-k'\))-face of
\(Q\) hits the simplex when shot from the query point. It is obvious that, for
all \(t < k'\), some (\(d-t\))-face of \(Q\) hits the simplex, and that, for all
\(t \geq k'\), no (\(d-t\))-face of \(Q\) hits the simplex.
For the sake of simplicity,
we hereafter focus on solving the face-shooting problem when \(k'=k\),
thus ignoring the fact a simplex can be missed.
Because the obtained running time will be of the order of
\(\tilde{O}(n^{k'-1})\), the running time will be dominated by this case.

In order to reduce face-shooting to range counting queries, we need an analogue of Lemma~\ref{lem:range} for convex combinations.
Let $A$ be a set of $k$ points $\{a_1,a_2,\ldots ,a_k\}$, and let $\Delta$ be a $(d-k+1)$-simplex with vertices in $B=\{b_1,b_2,\ldots ,b_{d-k+2}\}$.
We suppose that these points are in general position.
We define the hyperplanes
$H_i = \operatorname{aff}(A \cup B \setminus \{\,b_i, a_k\,\})$, for $i \in [d-k+2]$, and
$G_j = \operatorname{aff}(A \cup B \setminus \{\,a_j, a_k\,\})$, for $j \in [k-1]$.
We let \(H_i^+\) be the halfspace supported by \(H_i\) that contains \(b_i\), and
\(G_j^-\) the halfspace supported by \(G_j\) that does not contain \(a_j\).

\begin{lemma}
\label{lem:convex}
\[
\operatorname{conv}(A) \cap \Delta \neq \emptyset
\iff
a_k \in \left(
\left(\bigcap_{i=1}^{d-k+2} H_i^+\right)
\cap
\left(\bigcap_{j=1}^{k-1} G_j^-\right)
\right).
\]
\end{lemma}
\begin{proof}
$(\Leftarrow)$
Suppose that $a_k\in(\bigcap_i H^+_i)\cap(\bigcap_j G_j^-)$.
We have that $\operatorname{conv}(A)\cap\Delta\neq\emptyset$ if and only if both
$\operatorname{aff}(A)\cap\Delta\neq\emptyset$ and $\operatorname{conv}(A)\cap\operatorname{aff}(\Delta)\neq\emptyset$ hold.
From Lemma~\ref{lem:range}, we already have $\operatorname{aff}(A)\cap \Delta\not=\emptyset$.
It therefore remains to show that $\operatorname{conv}(A)\cap\operatorname{aff}(\Delta)\not=\emptyset$.

We first prove that $(\bigcap_j G_j)\cap\operatorname{conv} (A)\neq\emptyset$. We proceed by induction on $k$.
It can easily be shown to hold for $k=2$. Let us suppose it holds for $k-1$, and prove it for $k$.
The hyperplane $G_{k-1}$ separates $a_{k-1}$ from $a_k$.
Consider the point $a'_{k-1}$ of the segment between $a_{k-1}$ and $a_k$ that lies on $G_{k-1}$.
Let $A'= \{a_1,a_2,\ldots ,a_{k-2}, a'_{k-1}\}$.
Consider the intersection $G'_j$ of all hyperplanes $G_j$ for $j\in [k-2]$ with the subspace $\operatorname{aff}(A')$.
In the subspace $\operatorname{aff}(A')$, The hyperplanes $G'_j$ for $j\in [k-2]$ all separate $a_j$ from $a'_{k-1}$.
Hence we can apply induction on $A'$ and the hyperplanes $G'$ in dimension $k-2$, and we have that
$(\cap_{j\in [k-2]} G'_j)\cap\operatorname{conv} (A')\neq\emptyset$. Now because $a'_{k-1}\in\operatorname{conv} (\{a_{k-1},a_k\})$, we also have that
$(\cap_{j\in [k-1]} G_j)\cap\operatorname{conv} (A)\neq\emptyset$.

Now we also observe that $\bigcap_j G_j = \operatorname{aff} (\Delta)$.
The fact that \(\operatorname{aff} (\Delta) \subseteq \bigcap_j G_j\) is immediate since each
\(G_j\) contains \(\operatorname{aff}(\Delta)\).
To prove that \(\bigcap_j G_j\) cannot contain more than \(\operatorname{aff}(\Delta)\) it
suffices to show that those flats are of the same dimensions.
Since the set \(A \cup B\) is in general position, \(a_j\) (and \(a_k\)) cannot
lie on \(G_j\). Then we claim that the \(G_j\) are in general position. Indeed
if they are not, then there must be some \(1 \leq k' \leq k-1\) where \(\cap_{j \le k'-1} G_j =
\cap_{j \le k'} G_j\). However, this is not possible since \(a_{k'} \in \cap_{j \le
k'-1} G_j\) but \(a_{k'} \not\in \cap_{j \le k'} G_j\). The dimension of
\(\bigcap_j G_j\) is thus \(d-k+1\), the same as the dimension of \(\operatorname{aff}(\Delta)\).

Therefore, $\operatorname{conv}(A)\cap\operatorname{aff}(\Delta)\not=\emptyset$, as needed.\\

\noindent $(\Rightarrow)$
Suppose that $a_k\not\in(\bigcap_i H^+_i)\cap(\bigcap_j G_j^-)$.
Then one of the halfspace does not contain $a_k$.
It can be of the form $H^+_i$ or $G_j^-$.
In both cases, all points of $A$ are either contained in the hyperplane $H_i$ or $G_j$, or lie in $H^-_i$ or $G^-_j$.
Hence the hyperplane $H_i$ or $G_j$ separates the interiors of the convex hulls.
From the general position assumption, it also separates the convex hulls.
\end{proof}

The Lemma is illustrated on Figures~\ref{fig:convex} and~\ref{fig:convex2} in the cases $d=3$, $k=2$, and $d=k=3$.
\begin{figure}
\begin{center}
\includegraphics[page=3,scale=.8]{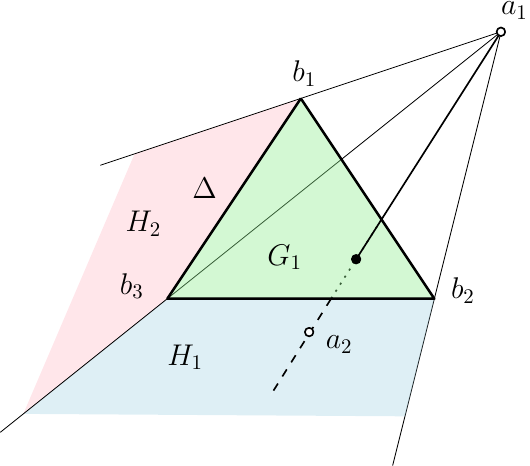}
\end{center}
\caption{\label{fig:convex}Illustration of Lemma~\ref{lem:convex} in the case $d=3$ and $k=2$. The segment $a_1a_2$ intersects $\Delta$ if and only if $a_2$ is located in the colored region below $\Delta$.}
\end{figure}

\begin{figure}
\begin{center}
\includegraphics[page=4,scale=.8]{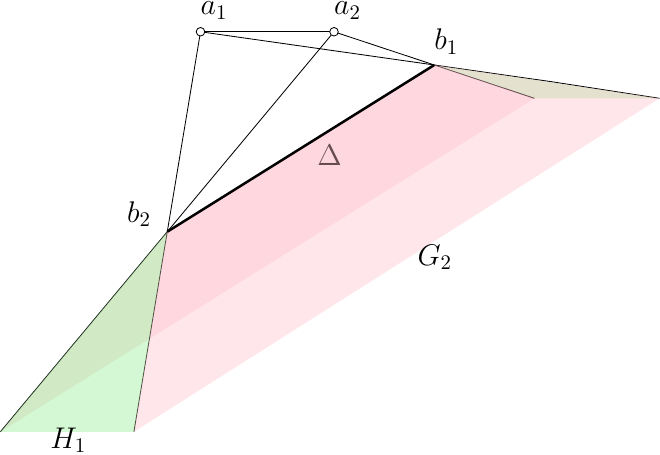}
\end{center}
\caption{\label{fig:convex2}Illustration of Lemma~\ref{lem:convex} in the case $k=d=3$. The triangle $a_1a_2a_3$ intersects $\Delta$ if and only if $a_3$ is located in the colored region.}
\end{figure}

\begin{proof}[Proof of Theorem~\ref{thm:convex}]
The algorithm follows the same steps as the algorithm described in the proof of Theorem~\ref{thm:flats}, except that the ranges used in the orthogonal range counting data structure are different, and involve a higher-dimensional space.

We reduce the problem to that of counting the number of $k$-subsets $A$ of $S$ whose convex hull intersects a given $(d-k+1)$-simplex $\Delta$.
We already argued that when fixing the first $k-2$ points $a_1,a_2,\ldots ,a_{k-2}$, the hyperplanes $H_i$ induce a circular order on the points of $S$.
Similarly, when the points $a_1,a_2,\ldots ,a_{k-2}$ are fixed, the hyperplanes $G_j$ all contain the $(d-2)$-flat $\operatorname{aff} (A\cup B\setminus \{a_j,a_{k-1},a_k\})$, hence also induce a circular order on the points of $S$.
Thus for each $(k-2)$-subset of $S$, we can assign $(d-k+2)+(k-1)=d+1$
coordinates to each point of $S$, one for each family of hyperplanes.
We then build an orthogonal range query data structure using these coordinates.
For each point $a_{k-1}$, we query this data structure and count the number of points $a_k$ such that $a_k\in (\bigcap_i H_i^+) \cap (\bigcap_j G_j^-)$.
From Lemma~\ref{lem:convex}, we can deduce the number of subsets $A$ whose convex hull intersects $\Delta$.

We can decrease by one the dimensionality of the ranges by realizing that the
supporting hyperplane of $G_{k-1}^-$ is unique as it does not depend on
$a_{k-1}$, only the orientation of $G_{k-1}^-$ does. To only output points
\(a_k\) such that \(a_k \in G_{k-1}^-\) we construct two data structures: one
with the points above \(G_{k-1}\) and one with the points below \(G_{k-1}\). We
query the relevant data structure depending whether \(a_{k-1}\) is above or
below \(G_{k-1}\).
This spares a logarithmic factor and yields an overall running time of
$O(n^{k-1}\log^{d-1}n)$ for the counting problem. Multiplying by the
\(O(\log n)\) rounds of binary search yields the claimed result.
\end{proof}

\section{Discussion}\label{sec:discussion}

We put the complexity of sparse regression on
par with that of degeneracy testing in the fixed-dimensional setting, and
proposed the first $o(n^k)$ algorithm for this problem.
A clear practical advantage of the
structure proposed by Har-Peled et al.~\cite{HIM18} is that one can reuse
approximate nearest neighbor data structures that circumvent the curse of
dimensionality, such as those based on locality-sensitive hashing~\cite{AI08}.
Our technique based on approximating the Euclidean distance by a polyhedral
distance, on the other hand, incurs an exponential dependence on $d$.

Some marginal improvements can probably be obtained by exploiting the
structure of the problem further.
For instance, we really only need orthogonal range {\color{darkred} emptiness queries} in
order to solve the decision version of our problem
Also, we do not harness the power of the Word RAM model for speeding
up the search~\cite{CLP11,CW16}.
For other relevant techniques, we refer to the excellent presentation by
Timothy Chan at WADS'13 on ``the art of shaving logs''\cite{Cha13}.

\subsection*{Acknowledgments}
The authors wish to thank the reviewers of a preliminary version of this manuscript,
who provided useful comments, as well as John Iacono and Stefan Langerman for
insightful discussions.\\

\bibliographystyle{abbrv}

\bibliography{paper}

\end{document}